\documentclass[aps,pra,10pt,twocolumn,superscriptaddress,floatfix,nofootinbib,showpacs,longbibliography]{revtex4-2}
\usepackage{makecell}
\newcolumntype{C}[1]{>{\centering\arraybackslash}m{#1}}
\usepackage{booktabs,adjustbox}
\usepackage{float}
\usepackage[normalem]{ulem}
\usepackage[utf8]{inputenc}  
\usepackage[T1]{fontenc}     
\usepackage[table]{xcolor}
\usepackage{color,soul}
\usepackage[british]{babel}  
\usepackage[sc,osf]{mathpazo}
\usepackage[scaled=0.86]{berasans}  
\usepackage[colorlinks=true, citecolor=blue, urlcolor=blue]{hyperref}  
\usepackage{graphicx} 
\usepackage[babel]{microtype}  
\usepackage{amsmath,amssymb,amsthm,bm,amsfonts,mathrsfs,bbm} 

\usepackage{xspace}  
\usepackage{pgf,tikz}
\usepackage{xcolor}
\usepackage{multirow}
\usepackage{array}
\usepackage{bigstrut}
\usepackage{braket}
\usepackage{color}
\usepackage{natbib}
\usepackage{multirow}
\usepackage{mathtools}
\usepackage{float}
\usepackage{blkarray}
\usepackage[caption = false]{subfig}
\usepackage{xcolor,colortbl}
\usepackage{color}
\usepackage{rotating}
\usepackage{tikz}
\usepackage{tabularx}

\newcommand{\Tr}{\operatorname{Tr}}

\newcommand{\be}{\begin{equation}}
\newcommand{\ee}{\end{equation}}
\newcommand{\ba}{\begin{eqnarray}}
\newcommand{\ea}{\end{eqnarray}}
\newcommand{\ketbra}[2]{|#1\rangle \langle #2|}
\newcommand{\tr}{\operatorname{Tr}}

\newtheorem{theorem}{Theorem}

\newtheorem{definition}{Definition}
\newtheorem{proposition}{Proposition}

\newtheorem{example}{Example}

\newtheorem{lemma}{Lemma}

\def\>{\rangle}
\def\<{\langle}

\usepackage{centernot}
\usepackage{subfig}

\begin{document}

\title{
Minimal Help, Maximal Gain: Environmental Assistance Unlocks Encoding Strength} 
\author{Snehasish Roy Chowdhury}
\affiliation{Physics and Applied Mathematics Unit, Indian Statistical Institute, 203 B. T. Road, Kolkata 700108, India}
\author{Sutapa Saha}
\affiliation{QIC Group, Harish Chandra Research Institute, A CI of Homi Bhabha National Institute, Chhatnag Road, Jhunsi, Prayagraj 211019, India}
\author{Subhendu B. Ghosh}
\affiliation{Department of Physics of Complex Systems, S. N. Bose National Centre for Basic Sciences, Block JD, Sector III, Salt Lake, Kolkata 700 106, India}
\author{Ranendu Adhikary}
\affiliation{Cryptology and Security Research Unit, Indian Statistical Institute, 203 B.T. Road, Kolkata 700108, India}
\author{Tamal Guha}
\affiliation{QICI Quantum Information and Computation Initiative, School of Computing and Data Science, The University of Hong Kong, Pokfulam Road, Hong Kong}
\affiliation{Physics and Applied Mathematics Unit, Indian Statistical Institute, 203 B. T. Road, Kolkata 700108, India}
\begin{abstract}
 For any quantum transmission line, with smaller output dimension than its input, the number of classical symbols that can be reliably encoded is strictly suboptimal. In other words, if the channel outputs a lesser number of symbols than it intakes, then rest of the symbols eventually leak into the environment, during the transmission. Can these lost symbols be recovered with minimal help from the environment? While the standard notion of environment-assisted classical capacity fails to fully capture this scenario, we introduce a generalized framework to address this question. Using an elegant example, we first demonstrate that the encoding capability of a quantum channel can be optimally restored with a minimal assistance of environment, albeit possessing suboptimal capacity in the conventional sense. Remarkably, we further prove that even the strongest two-input-two-output non-signaling correlations between sender and receiver cannot substitute for this assistance. Finally, we characterize a class of quantum channels, in arbitrary dimensions, exhibiting a sharp separation between the conventional environment-assisted capacity and the true potential for unlocking their encoding strength.
\end{abstract}
\maketitle
\textit{Introduction.}-- Perfect inversion of quantum channels, there by recovering all the encoded information, is a pivotal problem in quantum information science. It bears deep implication in the context of error correction \cite{nielsen1997reversible, nielesen1998information} and error-free encryption of quantum states \cite{ambainis2000private, boykin2003optimal, nayak2006invertible}. Indeed, the condition of error-correction singles out isometries as the only class of potential invertible channels \cite{nielsen2001quantum}. However, most general quantum channels, i.e, \textit{completely positive trace-preserving} (CPTP) maps, can be realized as a joint unitary on the system along with an ancillary system (referred as \textit{environment}) \cite{stinespring1955positive}. Since reversible dynamics conserve information, this raises an intriguing possibility of retrieving encoded information via nontrivial access to the environment.

While an implicit access to the environment was initially encountered in context correlated quantum channels \cite{macchiavello2002entanglement} and subsequently in \cite{bowen2004quantum, ball2004exploiting}, the first concrete framework was introduced by Gregoratti and Werner \cite{gregorotti2003quantum}. In this seminal result, they have demonstrated that even limited access---specifically, classical communication from a so-called \textit{friendly} environment---can enable perfect recovery of information in a broad class of channels.  This observation gave rise to the concept of environment-assisted capacity, revealing that classical coordination with an external environment can significantly enhance the capabilities of quantum channels. Access to the environment further boosted subsequent investigations in context of error correction \cite{gregoratti2004quantum}, decoherece mitigation \cite{buscemi2005inverting}, reliable transmission of quantum states \cite{karumanchi2016quantum, oskouei2021capacities, harraz2022quantum, harraz2023high}, interference of multiple quantum channels \cite{chiribella2019quantum, guha2023quantum, lai2024quick, saha2025interference}, along with their real-life implementations \cite{banaszek2004experimental,pirandola2021environment, rubino2021experimental, wang2025passive}.

Conventionally, the environment assisted classical capacity (EACC) of a quantum channel is estimated in terms of mutual information (MI) between the random variable, encoded in quantum systems, and the output random variable, obtained by performing a suitable measurement on the receiver-environment joint quantum states \cite{hayden2005correcting, winter2005environment, karumanchi2016classical}. With the limited access to the environment, the question of estimating the MI then boils down to discrimination of a set of orthogonal, possibly entangled, bipartite quantum states under limited measurement settings. This, in turn, renders sub-optimality of EACC for a channel, whenever such a discrimination of output states is not possible \cite{duan2009distinguishability, watrous2005bipartite}.

Here, we reformulate the notion of environment assisted classical communication from the perspectives of generalized information processing -- a framework of growing foundational and practical importance \cite{frenkel2015classical, patra2024classical, heinossari2024maximal, chiribella2025communication}. Within this framework numerous measures of classical communication, beyond mutual information, may emerge. In particular, we deal with a physically motivated one, corresponding to the cardinality of maximum numbers of classical symbols, transmitted reliably through a quantum channel. A simple yet crucial observation concerning this particular measure is that the classical communication under the assistance of environment empowers the reliable classical symbol transmission ability of the sender. More precisely, for quantum channels with higher input dimension (\(d_i\)) than the output one (\(d_o\)) the sender can reliably send only \(d_o\) number of symbols, while an assistance from environment can unlock the encoding strength further. We highlight the importance of this measure with an elegant example of a class of quantum channels possessing suboptimal EACC, yet their input-output statistics can only be simulated by a perfect quantum system of dimension identical to that of the channels' input. This concludes an optimal unlocking of encoding strength for these channels, however impossible to reveal with the conventional measure of EACC. More strikingly, instead of any assistance from environment, it is impossible to unlock the encoding strength optimally, even if the sender is allowed to additionally share any 2-input-2-output non-signaling correlations with the receiver. This, in turn, highlights an intriguing feature in context of communication through generalized broadcasting channels in quantum network. Finally, we extend our result for a class of quantum channels of arbitrary dimension, with suboptimal EACC, albeit depicting optimal unlocking of encoding strength with a minimal assistance from the environment.  
 
\textit{Generalized classical communication.}-- Consider a scenario where Alice and Bob two distant parties are connected via a quantum channel $\mathcal{N}:\mathcal{L}(\mathbb{C}^{d_A})\mapsto\mathcal{L}(\mathbb{C}^{d_B})$, where $\mathcal{L}(\mathbb{C}^{d_k})$ is the set of all linear operators acting on the complex Hilbert space \(\mathbb{C}^{d_k}\). For a given set of random variables $X=\{x_1,x_2,\cdots,x_n\}$, Alice encodes her information in the quantum states $\{\rho_i\}_{i=1}^n\in\mathcal{L}(\mathbb{C}^{d_A})$ and communicate to Bob through the channel $\mathcal{N}$. In each run, Bob is allowed to perform a generalized quantum measurement $\{\Lambda_k\}_{k=1}^m$ and accordingly extract the random variable $Y=\{y_1,y_2,\cdots,y_m\}$. The conditional probability over the random variables $Y$ given $X$ is then quantified as $p(y_j|x_i)=\text{Tr}[\Lambda_j\mathcal{N}(\rho_i)]$. Conventionally, the strength of classical communication in such a scenario is measured in terms of mutual information between $X$ and $Y$, given by
\begin{equation}\label{mi}
    \mathcal{C}^{(1)}(\mathcal{N})=\max_{\{p_x,\rho_x,\{\Lambda_y\}\}}I(X:Y),
\end{equation}
where, $\{p_x\}_{x=1}^n$ denotes the probability of the input random variable $X=\{x_i\}_{i=1}^n$. On the other hand,  collective decoding strategies further enhances the classical capacity of quantum channels to its Holevo quantity \cite{holevo1998capacity, schumacher1997sending}.

A more general approach towards classical communication considers the set of conditional probabilities $\{p(y_j|x_i)\}_{i=1,\dots,n}^{j=1,\dots,m}$, organized as an $n \times m$ row-stochastic matrix $[M(n,m)]_{i,j} = p(y_j|x_i)$, referred to as the \textit{channel matrix}. The set of such matrices achievable by transmitting isolated quantum states from $\mathcal{L}(\mathbb{C}^d)$ through a quantum channel $\mathcal{N}$ is denoted by $\mathcal{P}^{n\to m}(\mathcal{N}(\mathcal{Q}_d))$. For brevity, we denote $\mathcal{P}^{n\to m}(\mathcal{Q}_d)$ when $\mathcal{N}$ is the quantum identity channel. At this point, a key structural property of these sets becomes relevant.

\begin{proposition}\label{prop1}
    For any quantum channel \(\mathcal{N}:\mathcal{L}(\mathbb{C}^{d_A})\mapsto\mathcal{L}(\mathbb{C}^{d_B})\), and for all \(n,m\in\mathbb{N}\); \(\mathcal{P}^{n\to m}(\mathcal{N}(\mathcal{Q}_{d_A}))\subseteq \mathcal{P}^{n\to m}(\mathcal{Q}_d)\), where \(d=\min\{d_A,d_B\}\). 
\end{proposition}

While we defer the detailed proof to Appendix \ref{a1}, it is instructive to outline the core insight behind Proposition \ref{prop1}. The Proposition asserts that whatever input output statistics can be generated using the channel \(\mathcal{N}:\mathcal{L}(\mathbb{C}^{d_A})\mapsto\mathcal{L}(\mathbb{C}^{d_B})\), one can sufficiently simulate all of them by communicating an isolated \(d\)-level quantum system via a perfect identity channel. In other words, in classical communication setup, any such channel \(\mathcal{N}\) offers no advantage over quantum \(d\)-level identity channel.

On the other hand, any channel matrix $[M(n,m)]_{i,j} = p(y_j|x_i)$ can be achieved via communicating a minimum of $d$-level isolated quantum system, if and only if its positive semi definite (psd) rank is $d$ \cite{lee2014upper,heinossari2024maximal}. Formally, it is defined as,
\begin{definition}\label{def0}
    \cite{gouveia2013lifts} The psd rank of a non-negative matrix \(M\) of order \(n\times m\), denoted as \(\operatorname{rank}_{\operatorname{psd}}(M)\), is the smallest integer \(r\) such that there exist two sets of \(r \times r\) positive semidefinite matrices \(\{R_i\}_{i=1}^n\) and \(\{C_j\}_{j=1}^m\), such that:
\[
M_{ij} = \Tr(R_i C_j), \quad \forall ~i,j.
\]
\end{definition}
This along with Proposition \ref{prop1} implies that for every \(P\in\mathcal{P}^{n\to m}(\mathcal{N}(\mathcal{Q}_{d_A}))\), we have \(\text{rank}_{\text{psd}}(P)\le d\).

\textit{Environment assisted classical communication.}-- While the structure of a quantum channel is limited only up to system description, in a broader picture it can be visualized as an isometry from the operators acting on Hilbert space of the input system to that of the joint Hilbert space of the output system and the environment. Therefore, for any channel $\mathcal{N}:\mathcal{L}(\mathbb{C}^{d_A})\mapsto\mathcal{L}(\mathbb{C}^{d_B})$ between Alice and Bob, we can associate an isometry $\mathcal{V}_{\mathcal{N}}:\mathbb{C}^{d_A}\mapsto\mathbb{C}^{d_B}\otimes\mathbb{C}^{d_E}$, where $\mathbb{C}^{d_E}$ is the complex Hilbert space associated with the environment. The action of the isometry is connected with the action of the quantum channel in the following way:
\begin{small}
\begin{align*}
    \mathcal{N}(\rho_A)=\sigma_B\Leftrightarrow\mathcal{V}_{\mathcal{N}}(\rho_A)\mathcal{V}_{\mathcal{N}}^\dagger=\ket{\psi_{\sigma}}_{BE}\bra{\psi_{\sigma}},~\forall~\rho_A\in\mathcal{L}(\mathbb{C}^{d_A})
\end{align*}
\end{small}
such that $\text{Tr}_E(\ketbra{\psi_{\sigma}}{\psi_{\sigma}})=\sigma_B$.

Within this description, the degree of assistance from the environment can be characterized in various ways, depending upon the restrictions imposed on the measurement performed by the receiver and the environment \cite{winter2005environment}. In this context, we identify our scenario as the one with \textit{minimal} assistance, since the implementation of the decoding measurement does not require any additional communication between them. (Refer to the Appendix \ref{a2} for formal definitions.)

In any form of environment assistance for a quantum channel \(\mathcal{N}:\mathcal{L}(\mathbb{C}^{d_A})\mapsto\mathcal{L}(\mathbb{C}^{d_B})\), the set of modified channel matrices
\(\mathcal{P}_{\text{env}}^{n\to m}(\mathcal{N}(\mathcal{Q}_{d_A}))\), satisfies a trivial inclusion relation: \(\mathcal{P}^{n\to m}(\mathcal{N}(\mathcal{Q}_{d_A}))\subseteq \mathcal{P}_{\text{env}}^{n\to m}(\mathcal{N}(\mathcal{Q}_{d_A}))\). Environmental assistance is said to enhance communication utility if there exists at least one matrix $P \in \mathcal{P}^{n \to m}_{\text{env}}(\mathcal{N}(\mathcal{Q}_{d_A}))$ such that $P \notin \mathcal{P}^{n \to m}(\mathcal{N}(\mathcal{Q}_{d_A}))$. In such cases, some measure of classical communication—not necessarily mutual information—must increase under assistance (for details see appendix). For channels with $d_A > d_B$, however, environmental assistance may reveal a more fundamental benefit: enabling access to the full input encoding capacity. This advantage is not always captured by a specific capacity measure. Rather, a more generalized formalism of classical information processing, in terms of the input-output matrices, captures this notion. This leads to the following definition:
\begin{definition}\label{def1}
    Consider a channel \(\mathcal{N}:\mathcal{L}(\mathbb{C}^{d_A})\mapsto\mathcal{L}(\mathbb{C}^{d_B})\) with \(d_A > d_B\) , and \(d'\leq d_B\) is the minimum dimension of a quantum system, such that, \(\mathcal{P}^{n\to m}(\mathcal{N}(\mathcal{Q}_{d_A}))\subseteq \mathcal{P}^{n\to m}(\mathcal{Q}_{d'})\) for all \(n,m\in\mathbb{N}\). Environment assistance is said to unlock the encoding strength for \(\mathcal{N}\), whenever there exists a \(P\in\mathcal{P}_{\operatorname{env}}^{n^\prime\to m^\prime}(\mathcal{N}(\mathcal{Q}_{d_A}))\) but \(P\notin\mathcal{P}^{n^\prime\to m^\prime}(\mathcal{Q}_{d'})\). In other words, \(\operatorname{rank}_{\operatorname{psd}}(P)> d'\).
\end{definition}
\begin{figure}
    \centering
    \includegraphics[width=1.0\linewidth]{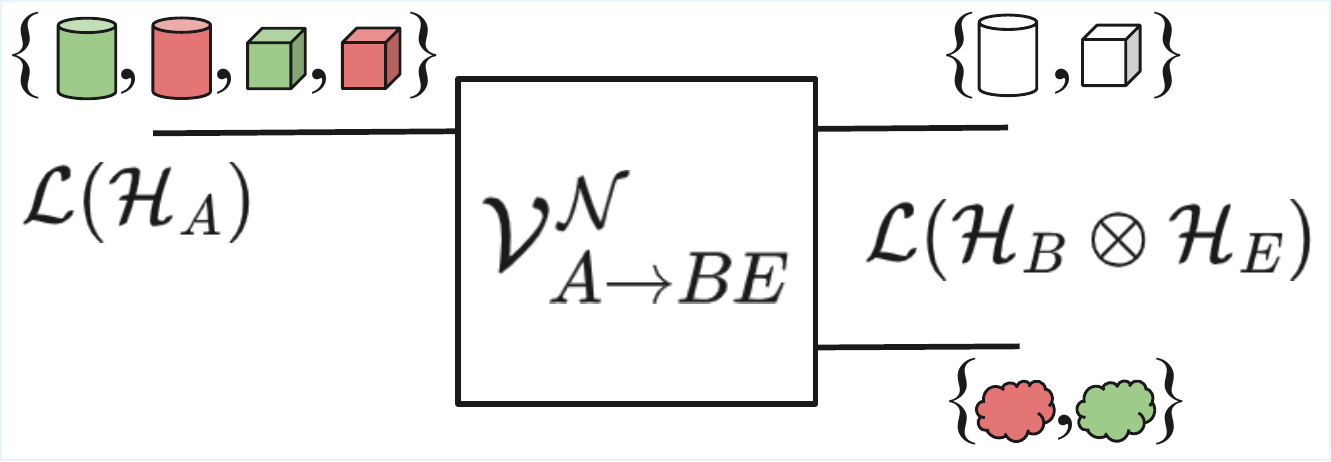}
    \caption{(Color online) A schematic diagram of unlocking the encoding strength with assistance of environment. Two binary information regarding the shape (\textit{cylinder/ cube}) and color (\textit{red/ green}) has been sent from Alice. Bob (environment) can only access the shape (color) information. Hence the channel \(\mathcal{N}_{A\to B}\), induced from the isometry \(\mathcal{V}^{\mathcal{N}}_{A\to BE}\), can only reliably communicate the shape information; However, with minimal assistance (the color information) from environment, Bob can reliably decode all the four possible information. Note that, while EACC captures the optimal unlocking of encoding strength here, this is not the case in general (see Theorem \ref{chanmatduan} and  \ref{gen}).}
    \label{fig0}
\end{figure}
With this definition, one can readily argue the following lemma:
\begin{lemma}\label{max}
    For any quantum channel $\mathcal{N}:\mathcal{L}(\mathcal{H}_A)\mapsto\mathcal{L}(\mathcal{H}_B)$ with \(d_A\ge d_B\), and for all \(n,m\in\mathbb{N}\), \(\mathcal{P}_{\operatorname{env}}^{n\to m}(\mathcal{N}(\mathcal{Q}_{d_A}))\subseteq \mathcal{P}^{n\to m}(\mathcal{Q}_{d_A})\), even with the complete access to the environment.
\end{lemma}
The proof follows directly from Proposition \ref{prop1}, however for sake of completeness, we detailed it in the Appendix \ref{a3}. With the help of Lemma \ref{max}, we can now formally define the optimal enhancement of encoding strength of a quantum channel under assistance of environment.
\begin{definition}\label{def2}
Any form of environment assistance is said to optimally unlock the encoding strength of \(\mathcal{N}:\mathcal{L}(\mathbb{C}^{d_A})\mapsto \mathcal{L}(\mathbb{C}^{d_B})\) (with \(d_A\geq d_B\)), if there exist a \(n,m\in\mathbb{N}\) such that \(\mathcal{P}_{\operatorname{env}}^{n\to m}(\mathcal{N}(\mathcal{Q}_{d_A}))\not\subseteq \mathcal{P}^{n\to m}(\mathcal{Q}_{\tilde d})\), for all \(\tilde d <d_A\).  
\end{definition}

The instance of achieving the optimal unlockable encoding strength, must be captured by a suitable measure of information processing, which need not to be the mutual information per se. Before revealing such phenomena, in the following we will characterize the scenario for which the traditional measure of EACC is sufficient to capture optimality (see Fig. \ref{fig0}), the proof of which is detailed in the Appendix \ref{a4}.     
\begin{lemma}\label{opt}
A quantum channel $\mathcal{N}:\mathcal{L}(\mathcal{H}_A)\mapsto\mathcal{L}(\mathcal{H}_B)$, reaches optimality of encoding strength, whenever its EACC is $\log d_A$. 
\end{lemma}
\noindent
This, in other words, states that for a quantum channel if (any form of) environment assistance unlocks its encoding strength optimally, then the generalized communication set-up can not further highlight any additional feature. In the following we will note few of such known quantum channels.
\begin{example}\label{ex1}
\cite{duan2009distinguishability,song2024existence} For any quantum channel (i) $\mathcal{N}:\mathcal{L}(\mathbb{C}^{d_A})\mapsto\mathcal{L}(\mathbb{C}^{d_B})$, with \(d_E(\le3)\) dimensional environment, and (ii) $\mathcal{N}:\mathcal{L}(\mathbb{C}^{d_A})\mapsto\mathcal{L}(\mathbb{C}^{d_B})$, where \(d_B\leq 3\), the encoding strength is unlocked optimally under the environment assistance. 
\end{example}
\noindent
Example \ref{ex1}, in line with Lemma \ref{opt}, demonstrate that the encoding strength is fully unlocked whenever the channel achieves its EACC. One might expect the converse to hold as well. Strikingly, this intuition fails. In the following, we construct a class of channels that defy this equivalence and thereby establish a genuine counterexample.  
\begin{lemma}\label{seventothree}
Consider an uncountable set of isometries \(\mathrm{S}\)
\begin{align}\label{duanspace}
\mathrm{S}:=\{\mathcal{V}_{7}&\mid\mathcal{V}_{7}:\mathbb{C}^7\mapsto\mathbb{C}^3\otimes\mathbb{C}^3\nonumber\\
&\&~\operatorname{Range}(\mathcal{V}_7)\perp\{\ket{\phi^+_3},\ket{i}\otimes\ket{j\neq i}\}
\end{align} 
where $\ket{\phi^+_3}=\frac{1}{\sqrt{3}}\sum_{k=0}^2\ket{k}\otimes\ket{k}$ and $i,j\in\{0,1,2\}$. 

EACC of all the channels \(\mathcal{N}^{\mathcal{V}}_7:\mathcal{L}(\mathbb{C}^7)\mapsto\mathcal{L}(\mathbb{C}^3)\), induced from each isometry \(\mathcal{V}_7\) is suboptimal, even if the decoding measurements of both the receiver and environment are separable super-operators (SEP).
\end{lemma}
\begin{proof}
    The proof follows from the fact that $\text{Range}(\mathcal{V}_7)$, for all $\mathcal{V}_7$, is indistinguishable under the SEP measurement implemented by the parties \cite{duan2009distinguishability}.

    Moreover, one can trivially argue using the set inclusion relation for different degrees of environment assistance that under minimal one the EACC of every \(\mathcal{N}^{\mathcal{V}}_7\) is strictly less than \(\log 7\)-bits.
\end{proof}
\noindent
In other words, in single shot regime no more than $6$ symbols can be sent perfectly through the class of channels $\mathcal{N}_7$. In the following theorem, we will establish that although EACC of these channels are suboptimal, in practice, these channels achieve optimality under minimal assistance of environment. 
\begin{theorem}\label{chanmatduan}
    All channels $\mathcal{N}^{\mathcal{V}}_7$ achieves optimal encoding strength with minimal assistance of environment.
\end{theorem}
\noindent
We detailed the proof in Appendix \ref{a6}, for which the following lemma will be instrumental. The proof of the lemma is also detailed in the Appendix \ref{a5}.  
\begin{lemma}\label{chanmatclas}
Consider a channel matrix $M_{7}(p)$, such that
\begin{eqnarray}\label{chanmat}
   \nonumber[M_7(p)]_{i,j}&=&\delta_{ij},\text{ when }i,j\leq5,\\
   \nonumber[M_7(p)]_{6,6}&=&p\text{ and }[M_7]_{6,7}=1-p\\
   \text{and }[M_7(p)]_{7,6}&=&\frac p3\text{ and }[M_7]_{7,7}=\frac{3-p}3
\end{eqnarray}
where, $0\leq p \leq 1$.
The PSD rank of \(M_7(p)\) is \(7\) for all \(p>0\).
\end{lemma}
\textit{From channel matrix to operational task.}--  
While at its face Theorem \ref{chanmatduan} looks like a mathematical artifact, every channel matrix has potential operational significance. More precisely, in a classical information processing task, the success pay-off can be interpreted as a function which maps the set of channel matrices to real numbers. As an illustration, consider the trace of a square channel matrix. Operationally this captures the accuracy with which a receiver can correctly identify the sender’s message; the more the trace, the higher the guessing probability of the receiver. Hence, the maximum trace over all the channel matrices generated by the channel \(\mathcal{N}\), can be interpreted as the \textit{Classical Transmission Fidelity} (\(\mathcal{F}_{c}\)) of the channel \(\mathcal{N}\), which reads
\[\mathcal{F}_{c}(\mathcal{N}):=\max_{n\in\mathbb{N}}~\max_{P\in\mathcal{P}^{n\to n}(\mathcal{N})}\tr[P]\] 
This readily gives an operational significance of Theorem \ref{chanmatduan}, which further obeys the following relation: 
\begin{proposition}\label{prop2}
   For every \(\mathcal{N}_7^{\mathcal{V}}\) channel, under the minimal assistance from environment, 
   \[\mathcal{F}_c^{env}(\mathcal{N}_7^{\mathcal{V}})>\mathcal{F}_c(\mathcal{Q}_6).\]
\end{proposition}
The proof is detailed in the Appendix \ref{a7}, which relies on the key observation that \(\mathcal{F}_c(\mathcal{Q}_6)=6\). However the trace of the matrix \(M_7\) reaches its maximum (\(=\frac {20}3>6\)) for \(p=1\), rendering a higher classical transmission fidelity of  \(\mathcal{N}_7^{\mathcal{V}}\) under minimal environment assistance (see Theorem \ref{chanmatduan}).

At this point it is worthwhile to mention that the mutual information (MI) of \(M_7(p=1)\), with assigning non-zero input probabilities over all the indices \(\{x_1\cdots, x_7\}\), can only reach up to \(\log 5\)-bits. Conversely, the maximum MI of \(M_7(p=1)\) can reach \(\log 6\)-bits at the cost of sacrificing the symbol \(x_7\). However, absence of that symbol results in a decreased trace. Hence, for the channel \(\mathcal{N}_7^{\mathcal{V}}\), under minimal assistance of environment, the strategy to maximize its \(\mathcal{F}_c\) is clearly distinct from that for conventional capacity enhancement, highlighting a stark inequivalence between them.
To underscore the power of environmental assistance, we compare it with preshared resources such as $2$-input--$2$-output non-signaling correlations and unlimited shared randomness. We show that, for all channels $\mathcal{N}^{\mathcal{V}}_7$, none of these can match even minimal environmental assistance in enhancing classical transmission fidelity (for proof see Appendix \ref{a8}).
\begin{theorem}\label{th2}
    \(\mathcal{F}_c^{env}(\mathcal{N}_7^\mathcal{V}) > \mathcal{F}_c^{PR}(\mathcal{N}_7^\mathcal{V})\ge\mathcal{F}_c^{SR}(\mathcal{N}_7^\mathcal{V})
    \) 
\end{theorem}

\textit{Generalization to arbitrary dimensions.}-- Standing as an example, Theorem \ref{chanmatduan} hence depicts that the general communication utility can violate the already established notion of sub-optimality. As an arbitrary generalization of our results, in the following we will first identify a class of quantum channels for which the conventional EACC is always suboptimal.
\begin{lemma}\label{d2-1tod}
For every \(d\ge 3\), consider an uncountable set of isometries \(\mathrm{S}_d\)
\begin{align}\label{duanspace}
\mathrm{S}_d:=\{\mathcal{V}_{d^2-1}&\mid\mathcal{V}_{d^2-1}:\mathbb{C}^{d^2-1}\mapsto\mathbb{C}^d\otimes\mathbb{C}^d\nonumber\\
&\&~\operatorname{Range}(\mathcal{V}_{d^2-1})\perp\ket{\phi^+_d}\}
\end{align} 
where $\ket{\phi^+_d}=\frac{1}{\sqrt{d}}\sum_{k=0}^{d-1}\ket{k}\otimes\ket{k}$.

EACC of all channels \(\mathcal{N}^{\mathcal{V}}_{d^2-1}:\mathcal{L}(\mathbb{C}^{d^2-1})\mapsto\mathcal{L}(\mathbb{C}^d)\), induced from each isometry \(\mathcal{V}_{d^2-1}\) is suboptimal, even if the decoding measurements of both the receiver and the environment are separable super-operators (SEP).
\end{lemma}
\begin{proof}
    The proof follows directly from the fact that no bases spanning the subspace orthogonal to $\ket{\phi_d^+}$ can be distinguished perfectly under LOCC \cite{watrous2005bipartite}.
    \end{proof}
    However, we will now show that for the same class of quantum channels, a minimal assistance from the environment to the receiver, empowers the sender to unlock its true encoding strength optimally. The proof of the same is depicted in the Appendix \ref{a9}.  
\begin{theorem}\label{gen}
    For all \(d\ge3\), quantum channels \(\mathcal{N}^{\mathcal{V}}_{d^2-1}\), corresponding to all \(\mathcal{V}_{d^2-1}\in\mathrm{S}_{d}\), achieve optimal encoding strength under minimal assistance of the environment.
\end{theorem}  
\textit{Discussions.}-- In summary, we have introduced the notion of unlocking the true classical information encoding strength for quantum channels, with input dimension exceeding the output. Our analysis, here in particular, is related to environment-assisted classical communication, i.e., by taking the limited access of environment as a resource. However, the concept serves as a general tool for characterizing communication advantage provided by arbitrary resources. This approach is particularly valuable in single-shot regime, where infinitely many inequivalent capacity measures can exist. Accordingly, to establish a resource induced communication advantage there, one would need to identify a relevant measure that increases under the given resource. Alternatively, if the resource unlocks encoding strength, such a measure is guaranteed to exist and hence quantified as a sufficient criterion for the resource to wit an enhanced communication utility. 

The situation is closely analogous to the resource theory of purity \cite{horodecki2003reversible}, particularly under single-shot \cite{gour2015resource}. There, the monotones are given by the entire family of R\'enyi entropies, and comparing the resource content of two states \(\rho\) and \(\sigma\) would, in principle, require comparing all R\'enyi entropies. However, if \(\rho\) majorizes \(\sigma\), then one can conclude immediately that \(\rho\) contains more purity, since \(\rho\) can be transformed to \(\sigma\) under noisy operations. In a similar spirit, unlocking of encoding strength provides a structural,``\textit{majorization}" alike criterion for identifying the communication advantage of any resource. While the explicit identification of a specific capacity measure, enhanced by that particular resource, carries a more operational interpretation. In our work, classical transmission fidelity serves as one such physically motivated measure.

Another highlighting feature of our work is that the advantages we uncover arise from \textit{minimal} assistance of environment, where the decoding measurements of receiver and environment do not require any prior classical communication between them. This is particularly interesting since only such assistance does not exploit measurement incompatibility---a distinctly nonclassical feature of quantum theory---yet unlocks encoding strength of the channel optimally. Furthermore we have shown that shared resources between the sender and the receiver---such as shared randomness or any \(2\)-\(2\)-\(2\)--nonlocal correlation---fail to match the power of minimal assistance from environment. This leads to a counterintuitive implication in the setting of quantum broadcast channel scenario \cite{yard2011quantum, bhattacharya2021prx}, particularly when one sender communicates simultaneously to two receivers through a channel and its complement. Our findings reveal that, in this scenario, communication between two receivers might over-perform some non-signaling correlation shared between sender and the receivers. The nontriviality of this observation lies in the fact that while the shared resources can be employed in the communication protocol depending on the classical message at the sender's end, communication between the receiver does not depend on that. Finally, we extend our result for arbitrary dimension by explicit construction of quantum channels.

\textit{Acknowledgments.}-- SRC acknowledges support from University Grants Commission, India (reference no. 211610113404). SBG acknowledges the financial support through the National Quantum Mission (NQM) of the Department of Science and Technology, Government of India. RA acknowledges financial support from the Council of Scientific and Industrial Research (CSIR), Government of India under File No.09/0093(19292)/2024-EMR-I. TG would like to acknowledge his academic visit at the Physics and Applied Mathematics Unit, Indian Statistical Institute, Kolkata, India.
\bibliography{bibliography}
\section*{Appendix}
\subsection{Proof of Proposition \ref{prop1}}\label{a1}
Let us first consider a quantum channel \(\mathcal{N}:\mathcal{L}(\mathcal{H}_{d_A})\mapsto\mathcal{L}(\mathcal{H}_{d_B})\), for which \(d_A\leq d_B\). Suppose, \(P\) is an arbitrary \(n\)-input-\(m\)-output channel matrix, obtained by sending the quantum states \(\{\rho_i\}_{i=1}^n\in\mathcal{L}(\mathcal{H}_{d_A})\) through the channel \(\mathcal{N}\) and performing a \(m\)-outcome POVM \(\{M_j|M_j\geq0,~\sum_jM_j=\mathbb{I}_{d_B}\}\) at the decoders end. Therefore, \(P\in\mathcal{P}^{n\to m}(\mathcal{N}(\mathcal{Q}_{d_A}))\) and assume, if possible \(P\notin\mathcal{P}^{n\to m}(Q_{d_A})\). 

Now, consider a perfect \(d_A\)-dimensional quantum channel \(\mathcal{I}_{d_A}:\mathcal{L}(\mathcal{H}_{d_A})\mapsto\mathcal{L}(\mathcal{H}_{d_B})\). Now, if the sender, Alice sends the same \(\{\rho_i\}_{i=1}^n\) states through \(\mathcal{I}_{d_A}\), then the receiver, Bob will get these states identically. Then before decoding if Bob applies the channel \(\mathcal{N}\) and perform the same POVM \(\{M_j\}_{j=1}^m\), then they can effectively generate the channel matrix \(P\). In other words, \(P\in\mathcal{P}^{n\to m}(\mathcal{Q}_{d_A})\). But this leads to a contradiction and hence, 
\begin{equation}\label{ep11}
    \mathcal{P}^{n\to m}(\mathcal{N}(\mathcal{Q}_{d_A}))\subseteq\mathcal{P}^{n\to m}(Q_{d_A}), \text{ when }d_A\leq d_B.
\end{equation}
We will now move to the other side of the statement, considering the same quantum channel \(\mathcal{N}\) with \(d_A>d_B\). Now, any \(n\times m\) channel matrix \(P\in\mathcal{P}^{n\to m}(\mathcal{N}(\mathcal{Q}_{d_A}))\), we have \[[P]_{ij}=\Tr[M_j\mathcal{N}(\rho_i)],\] where \(\rho_i\in\mathcal{L}(\mathcal{H}_A),~\mathcal{N}(\rho_i)\in\mathcal{L}(\mathcal{H}_B)\) and \(M_j\in\mathcal{L}(\mathcal{H}_B)\) is a POVM operator. In other words, every element of the channel matrix \(P\), can be realized as an inner product between two \(d_B\times d_B\) positive semi-definite matrices \(\{\mathcal{N}(\rho_i)\}_{i=1}^n\) and \(\{M_j\}_{j=1}^m\). Therefore, for any \(P\in\mathcal{P}^{n\to m}(\mathcal{N}(\mathcal{Q}_{d_A}))\), the positive semi-definite rank will satisfy \(\text{rank}_{psd}(P)\leq d_B\) \cite{fawzi2015positive, lee2014upper}. Hence the channel matrix \(P\) can be sufficiently simulated by communicating a \(d_B\)-dimensional quantum system \cite{heinossari2024maximal}. That is,   
\begin{equation}\label{ep12}
    P\in\mathcal{P}^{n\to m}(\mathcal{Q}_{d_B}), \text{ when }d_A>d_B.
\end{equation}
Eqs. (\ref{ep11}) and (\ref{ep12}) then readily implies 
\begin{equation}\label{ep13}
\mathcal{P}^{n\to m}(\mathcal{N}(\mathcal{Q}_{d_A}))\subseteq\mathcal{P}^{n\to m}(Q_d)\text{, where }d=\min\{d_A,d_B\}.\end{equation}
\subsection{Different degrees of environment assistance and classification of the sets \(\mathcal{P}^{n\to m}_{\operatorname{env}}(\mathcal{Q}_{d_A})\)}\label{a2}
Depending upon the causal constraints imposed on the decoding measurement performed by the receiver and the environment, we can assign a clear hierarchy on the degree of environment assistance. The notion of such degrees, from the perspective of conventional EACC measures, has been introduced in \cite{winter2005environment}. However, in the current context, we would like to revisit those aspects from the generalized communication perspective.

Let us first consider the scenario of \textit{environment-assisted} classical communication for the channel \(\mathcal{N}:\mathcal{L}(\mathcal{H}_A)\mapsto\mathcal{L}(\mathcal{H}_B)\), involving an one-way LOCC assistance from the environment to the receiver (Bob) to implement the decoding POVM. The set of all such \(n\times m\) channel matrix, generated under such an assistance, is denoted as \(\mathcal{P}^{n\to m}_{\leftarrow}(\mathcal{N}(\mathcal{Q}_{d_A}))\). Therefore, for any \(P\in\mathcal{P}^{n\to m}_{\leftarrow}(\mathcal{N}(\mathcal{Q}_{d_A})\), we have: 
\begin{align}\label{em1}
[P]_{ij}=\overleftarrow{p}(y_j|x_i)=\sum_k\Tr[(\Lambda_{j|k}^{B}\otimes\Lambda_k^{E})\mathcal{V}_{\mathcal{N}}(\rho_i)],
\end{align}
where $\{\Lambda_k^E\}$ is the decoding POVM applied on the environment system and $\{\Lambda_{j|k}^B\}_j$ is the POVM applied by Bob, depending on the $k^{th}$ outcome at the environment end. Additionally, \(\mathcal{V}_{\mathcal{N}}\) denotes the isometry corresponding to \(\mathcal{N}\) and \(\rho_i\in\mathcal{L}(\mathcal{H}_A)\) denotes Alice's preparation corresponding to the input \(x_i\).

Conversely, in the \textit{environment-assisting} case, implementation of decoding POVM requires a one-way communication from receiver to the environment. A channel matrix \(P\in\mathcal{P}^{n\to m}_{\rightarrow}(\mathcal{N}(\mathcal{Q}_{d_A})))\), where \(\mathcal{P}^{n\to m}_{\rightarrow}(\mathcal{N}(\mathcal{Q}_{d_A}))\) denotes the set of all \(n\times m\) channel matrices achievable via such simulation, is given by
\begin{align}\label{em2}
[P]_{ij}=\overrightarrow{p}(y_j|x_i)=\sum_{l}\sum_{k}\Tr[(\Lambda_{l}^{B}\otimes\Lambda_{k|l}^{E})\mathcal{V}_{\mathcal{N}}(\rho_i)],
\end{align}
where $j:=f(k,l)$ is decoding function over environment's outcome $k$ and Bob's outcome $l$. 

Similarly, one could consider the scenario of unbounded \textit{LOCC-}, \textit{SEP-} and \textit{PPT-assisted} classical communication from the environment, respectively by lifting the limitations on the set of allowed measurements. The corresponding set of \(n\times m\) channel matrices can be respectively denoted as \(\mathcal{P}^{n\to m}_{\text{LOCC}}(\mathcal{N}(\mathcal{Q}_{d_A})),~\mathcal{P}^{n\to m}_{\text{SEP}}(\mathcal{N}(\mathcal{Q}_{d_A}))\) and \(\mathcal{P}^{n\to m}_{\text{PPT}}(\mathcal{N}(\mathcal{Q}_{d_A}))\).

Finally, the \textit{minimal} assistance of environment is a scenario, where the implementation of decoding measurements, performed by both the receiver and environment, does not require any communication between the environment and Bob. We denote the set of all \(n\times m\) channel matrices in this context, as \(\mathcal{P}^{n\to m}_{\text{min}}(\mathcal{N}(\mathcal{Q}_{d_A}))\) and any matrix \(P\in\mathcal{P}^{n\to m}_{\text{min}}(\mathcal{N}(\mathcal{Q}_{d_A}))\) can be written as:
\begin{equation}\label{em3}
[P]_{ij}=p(y_j|x_i)=\Tr[(\Lambda_l^{B}\otimes\Lambda_k^{E})\mathcal{V}_{\mathcal{N}}(\rho_i)]
\end{equation}
where $j:=f(k,l)$ as defined earlier.

With all these degrees of environment assistance we can conclude a set inclusion relation (see Fig. \ref{fig1}) for every pair of \(n,m\). This, in turn, justifies the name \textit{minimal environment assistance} in the present work.  
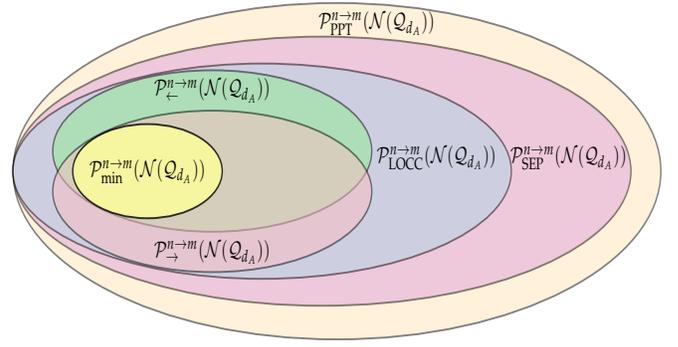
\begin{figure}[t!]
    \centering
    \resizebox{0.485\textwidth}{4.5cm}{%
    \begin{tikzpicture}

\definecolor{pptcol}{RGB}{255,228,181}   
\definecolor{sepcol}{RGB}{221,160,221}   
\definecolor{locccol}{RGB}{173,216,230}  
\definecolor{rightcol}{RGB}{255,190,193} 
\definecolor{leftcol}{RGB}{144,238,144}  
\definecolor{envcol}{RGB}{255,255,153}   

\filldraw[pptcol,opacity=0.5,draw=black,thick,rounded corners=25pt]
  (-6,0) ellipse (6.5 and 2.5);
\node at (-5.2,2.2) {$\mathcal{P}^{n\to m}_{\text{PPT}}(\mathcal{N}(\mathcal{Q}_{d_A}))$};

\filldraw[sepcol,opacity=0.5,draw=black,thick,rounded corners=25pt]
  (-6.3,0) ellipse (6.2 and 2.0);
\node at (-1.35,0.2) {$\mathcal{P}^{n\to m}_{\text{SEP}}(\mathcal{N}(\mathcal{Q}_{d_A}))$};

\filldraw[locccol,opacity=0.5,draw=black,thick,rounded corners=25pt]
  (-7.5,0) ellipse (5 and 1.6);
\node at (-4,0.2) {$\mathcal{P}^{n\to m}_{\text{LOCC}}(\mathcal{N}(\mathcal{Q}_{d_A}))$};

\filldraw[leftcol,opacity=0.5,draw=black,thick,rounded corners=20pt]
  (-8.5,0.3) ellipse (3.2 and 1.2);
\node at (-8.5,1.2) {$\mathcal{P}^{n\to m}_{\leftarrow}(\mathcal{N}(\mathcal{Q}_{d_A}))$};

\filldraw[rightcol,opacity=0.5,draw=black,thick,rounded corners=20pt]
  (-8.5,-0.3) ellipse (3.2 and 1.2);
\node at (-8.5,-1.2) {$\mathcal{P}^{n\to m}_{\to}(\mathcal{N}(\mathcal{Q}_{d_A}))$};

\filldraw[envcol,opacity=0.8,draw=black,thick,rounded corners=15pt]
  (-9.8,0) ellipse (1.5 and 0.7);
\node at (-9.8,0) {$\mathcal{P}^{n\to m}_{\min}(\mathcal{N}(\mathcal{Q}_{d_A}))$};

\end{tikzpicture}
}
    \caption{Classification of the sets \(\mathcal{P}^{n\to m}_{\text{env}}(\mathcal{Q}_{d_A})\) depending on various causal constraints imposed upon the decoding measurements of the receiver and environment.}\label{fig1}
    \label{fig:placeholder}
\end{figure}
\subsection{Proof of Lemma \ref{max}}\label{a3}
Every quantum channel \(\mathcal{N}:\mathcal{L}(\mathcal{H}_A)\mapsto\mathcal{L}(\mathcal{H}_B)\) can be associated with an isometry \(\mathcal{V}_{\mathcal{N}}:\mathcal{H}_A\mapsto\mathcal{H}_B\otimes\mathcal{H}_E\), where \(\mathcal{H}_E\) is the Hilbert space associated with the environment. Since, \(\mathcal{V}_{\mathcal{N}}^{\dagger}\mathcal{V}_{\mathcal{N}}=\mathbb{I}_{d_A}\), it is easy to identify \(\mathcal{V}_{\mathcal{N}}\) as an one-rank quantum channel itself: \[\mathcal{V}_{\mathcal{N}}X_A\mathcal{V}_{\mathcal{N}}^{\dagger} = X_{B'}, \text{ such that } X_k\in\mathcal{L}(\mathcal{H}_{k}),\] 
where \(\text{dim.}(\mathcal{H}_A)=\text{dim.}(\mathcal{H}_{B'})\leq\text{dim.}(\mathcal{H}_{B}\otimes\mathcal{H}_{E})\). Also note that, the limited assistance from environment restricts to perform any potential joint (possibly, entangled) measurement to perform on \(\mathcal{L}(\mathcal{H}_B\otimes\mathcal{H}_E)\simeq\mathcal{L}(\mathcal{H}_{B'})\). This further implies
\begin{equation}\label{el11}
\mathcal{P}^{n\to m}_{env}(\mathcal{N}(\mathcal{Q}_{d_A}))\subseteq\mathcal{P}^{n\to m}(\mathcal{V}_{\mathcal{N}}(\mathcal{Q}_{d_A})).\end{equation}
It is now trivial to argue from Eq. (\ref{ep13}), that
\begin{equation}\label{el12}
\mathcal{P}^{n\to m}(\mathcal{V}_{\mathcal{N}}(\mathcal{Q}_{d_A}))\subseteq\mathcal{P}^{n\to m}(\mathcal{Q}_{d_A}).
\end{equation}
Eq. (\ref{el11}) and (\ref{el12}), hence together implies
\[\mathcal{P}^{n\to m}_{env}(\mathcal{N}(\mathcal{Q}_{d_A}))\subseteq\mathcal{P}^{n\to m}(\mathcal{Q}_{d_A}).\]
\subsection{Proof of Lemma \ref{opt}}\label{a4}
Consider a quantum channel \(\mathcal{N}:\mathcal{L}(\mathbb{C}^{d_A})\mapsto\mathcal{L}(\mathbb{C}^{d_B})\) (possibly \(d_A\geq d_B\)) for which EACC is exactly \(\log d_A\). In other words, \(d_A\)-numbers of classical symbols can be reliably transmitted through the channel \(\mathcal{N}\), when assisted by the environment.

Therefore, \(P=\mathbb{I}_{d_A}\) is a \(d_A\times d_A\) channel matrix, such that \(P\in\mathcal{P}^{n\to m}_{env}(\mathcal{Q}_{d_A})\). At this point, let us consider a function \(\Lambda_{\max}\) (namely, the \textit{max-monotone}) for any non-negative matrix \(M\) of order \(r\times c\), such that 
\begin{equation}\label{el21}
\Lambda_{\max}(M)=\sum_{i=1}^r~\max_{j\in\{1,\cdots, c\}}M_{ij}.\end{equation}
Now, for every row-stochastic matrix \(M\), we have  \cite{lee2014upper}
\begin{equation}\label{el22}
\text{rank}_{\text{psd}}(M)\geq \Lambda_{\max}(M).
\end{equation} 
Therefore, by noting that \(\Lambda_{\max}(\mathbb{I}_{d_A})=d_A\), we can conclude 
\(\text{rank}_{\text{psd}}(\mathbb{I}_{d_A})\geq d_A.\)

On the other hand, with a \(d_A\)-dimensional perfect quantum channel one can trivially simulate \(\mathbb{I}_{d_A}\): Using the computational basis preparation \(\{\ket{i}\}_{i=0}^{d_A-1}\) and a computational basis measurement \(\{\ketbra{i}{i}\}_{i=0}^{d_A-1}\). Hence, \(\text{rank}_{\text{psd}}(\mathbb{I}_{d_A})\leq d_A,\)
which further implies 
\begin{equation}\label{el23}
\text{rank}_{\text{psd}}(\mathbb{I}_{d_A})=d_A.
\end{equation}
Therefore, \(P\notin\mathcal{P}^{n\to m}(\mathcal{Q}_{\tilde{d}})\), where \(\tilde{d}<d_A\). That is, in general
\[\mathcal{P}_{env}^{n\to m}(\mathcal{N}(\mathcal{Q}_{d_A}))\nsubseteq\mathcal{P}^{n\to m}(\mathcal{Q}_{\tilde{d}}),\quad\forall~\tilde{d}<d_A.\]
\subsection{Two instrumental Lemmas for PSD rank}
\begin{lemma}\label{sl1}\cite{fawzi2015positive}
Consider a non-negative block matrix \(\mathbb{M}\) of order \(r\times c\), such that
\begin{align*}
    \mathbb{M}=\begin{pmatrix}
      \mathbb{A}&\Theta\\
      \mathbb{B}&\mathbb{D}
    \end{pmatrix},
\end{align*}
where, \(\mathbb{A},~\mathbb{B}\text{ and }\mathbb{D}\) are the non-negative matrices of orders \(r_1\times c_1,~r_2\times c_1\text{ and }r_2\times c_2\) respectively and \(\Theta\) is a null matrix of order \(r_1\times c_2\). Also, \(r_1+r_2=r\) and \(c_1+c_2=c\). Then,
\begin{equation}\label{esl11}
    \text{rank}_{\text{psd}}(\mathbb{M})\geq \text{rank}_{\text{psd}}(\mathbb{A})+\text{rank}_{\text{psd}}(\mathbb{D}).
\end{equation}
Moreover the equality holds when \(\mathbb{B}\) is also a null matrix. 
\end{lemma} 
\begin{proof}
    The statement is proved in the Theorem 2.10 of the Ref. \cite{fawzi2015positive}. 
\end{proof}
\begin{lemma}\label{sl2}
    For any non-negative triangular matrix, the psd rank is lower-bounded by the order of the matrix. The bound further saturates if the matrix is diagonal. 
\end{lemma}
\begin{proof}
    First note that, for any non-negative matrix \(\mathbb{M}\), \[\text{rank}_{\text{psd}}(\mathbb{M})=\text{rank}_{\text{psd}}(\mathbb{M}^T),\]
    where, \(*^T\) denotes the transposition operation.
    Therefore, in the following we will prove our result for lower-triangular matrices only and the same also holds for the upper-triangular one. 

    Let us now consider a non-negative lower-triangular matrix \(\mathbb{L}\) of order \(n\times n\). Using Eq.(\ref{esl11}), we can then write
    \begin{equation}\label{esl21}
    \text{rank}_{\text{psd}}(\mathbb{L})\geq \text{rank}_{\text{psd}}([l_{11}])+\text{rank}_{\text{psd}}(\mathbb{L}_1),
    \end{equation}
    where \(l_{11}=[\mathbb{L}]_{1,1}\) and \(\mathbb{L}_1\) is a \((n-1)\times(n-1)\) lower-triangular matrix omitting the first row and first column of the matrix \(\mathbb{L}\). 

    Similarly, we can recursively construct a set of  lower-triangular matrices \(\{\mathbb{L}_2,\mathbb{L}_3,\cdots,\mathbb{L}_{n-1}\}\), such that \(\mathbb{L}_k\) is generated by  omitting the first row and first column of the matrix \(\mathbb{L}_{k-1}\). Hence, we can rewrite the Eq.(\ref{esl21}) as 
    \[\text{rank}_{\text{psd}}(\mathbb{L})\geq \sum_{k=1}^{n-1} \text{rank}_{\text{psd}}([l_{kk}]) + \text{rank}_{\text{psd}}(\mathbb{L}_{n-1}).\]
    Finally, by noting that the matrix \(\mathbb{L}_{n-1}=[l_{nn}]\) and that the psd rank of any \(1\times 1\) matrix (scalar) is trivially 1, we have
    \[\text{rank}_{\text{psd}}(\mathbb{L})\geq n.\]
    This completes the proof. 
    
    With the help of Lemma \ref{sl1}, it is now trivial to argue that the bound saturates for \(\mathbb{L}\) being diagonal.
    
\end{proof}
\subsection{Proof of Lemma \ref{chanmatclas}}\label{a5}
Let us first show that \(\text{rank}_{\text{psd}}(M_7(p))\leq 7\) $\forall$ \(p>0\). To this goal, the following strategy establishes that a \textit{seven}-dimensional lone quantum system is sufficient to simulate all the channel matrices \(M_7(p)\).
\begin{enumerate}
\item Alice uses the \textit{seven}-dimensional computational basis \(\{\ket{i}\}_{i=0}^6\) to encode her input random variables \(X:=\{x_0,\cdots,x_6\}\) respectively. She sends the quantum system reliably through a perfect \textit{seven}-dimensional quantum channel to Bob.
\item Upon receiving the system Bob performs the computational basis measurement \(\{P_j:=\ketbra{j}{j}\}_{j=0}^6\).
\item He then outputs the random variable \(\{y_j\}_{j=0}^4\) whenever the projector \(\{P_j\}_{j=0}^4\) clicks. Additionally, after getting the click of the \(P_5\) (\(P_6\)) projector, he uses a \(\{p,1-p\}\) (\(\{\frac p3,1-\frac p3\}\)) local randomness to output \(y_5\) and \(y_6\) respectively. 
\end{enumerate}
To prove the converse, that is \(\text{rank}_{\text{psd}}(M_7(p))\geq 7\), first note that the channel matrix 
\[M_7(p)=\mathbb{I}_5\oplus\mathbb{P}(p)\text{,  where }\mathbb{P}(p)=\begin{pmatrix}
    p & 1-p\\
    \frac p3 & 1-\frac p3
\end{pmatrix}.\]
Therefore, using Lemma \ref{sl1}, we can write
\begin{align}
  \nonumber  \text{rank}_{\text{psd}}(M_7(p))&=\text{rank}_{\text{psd}}(\mathbb{I}_5)+\text{rank}_{\text{psd}}(\mathbb{P}(p))\\\nonumber
    &\geq 5 + \Lambda_{\max}(\mathbb{P}(p))\\\nonumber
    &=5+\max\{p,1-p\}+1-\frac p3\\\nonumber
    &>6.
\end{align}
The first inequality follows from Eq. (\ref{el22}) and (\ref{el23}). The second equality follows from  Eq. (\ref{el21}) and using the fact that \(\frac p3< \frac 12\) whenever \(0\leq p\leq 1\). Finally, the last strict inequality can argued trivially: \(\max\{p,1-p\}\geq\frac 12 > \frac p3\). This, in other words, implies that a perfect quantum system of dimension \textit{six} is unable to simulate the channel matrix \(M_7\). Hence,
\[\text{rank}_{\text{psd}}(M_7(p))=7\quad\forall~p>0.\]
\subsection{Proof of Theorem \ref{chanmatduan}}\label{a6}
We will prove the theorem by showing that the channels \(\mathcal{N}^{\mathcal{V}}_7\), corresponding to all isometries \(\mathcal{V}_7\in\mathrm{S}\), can simulate \(M_7(p)\) under minimal assistance of environment. To this end, consider an orthonormal basis of \(\text{Range}(\mathcal{V}_7)\):
 \begin{align*}
   \left\{\!\begin{aligned}\ket{\psi_1}=\ket{02},~\ket{\psi_2}=\ket{10},~\ket{\psi_3}=\ket{12},\\\ket{\psi_4}=\ket{20},~\ket{\psi_5}=\ket{21},~
      \ket{\psi_6}=\frac{1}{\sqrt{2}}(\ket{00}-\ket{11})\\ \text{and } ~\ket{\psi_7}={\frac{1}{\sqrt{6}}(\ket{00}+\ket{11})-\sqrt\frac{2}{3}}\ket{22}\end{aligned}\right\}\end{align*}
      Notice that, for all \(\mathcal{V}_7\in\mathrm{S}\), there exist a basis \(\mathbf{B}_{\mathcal{V}}\) of \(\mathbb{C}^7\) such that:
      \begin{align}\label{step2}
          \mathcal{V}_7\ket{\phi_i}=\ket{\psi_i}\quad\forall~\ket{\phi_i}\in\mathbf{B}_{\mathcal{V}}
      \end{align}
Now, a strategy to simulate the channel matrix \(M_7(p)\) using the channel \(\mathcal{N}^\mathcal{V}_7\) is listed below:

{\bf Encoding.} Depending on the isometry \(\mathcal{V}_7\), Alice encodes her random variable via the map \(\mathcal{E}:X\mapsto\mathbf{B}_\mathcal{V}\) defined as \(\mathcal{E}(x_i)=\ket{\phi_i}\) for all \(i\) and sends it via the corresponding channel \(\mathcal{N}^{\mathcal{V}}_7\).
\begin{widetext}
\begin{center}
\begin{table}[htb!]
    \centering
    \resizebox{\textwidth}{!}{
    \begin{tabular}{|l|c||c|c|c||l|}
        \hline
        \multicolumn{2}{|c||}{Encoding} & \multicolumn{3}{c||}{Decoding} & \multicolumn{1}{c|}{Input-Output Probability} \\ \hline
        Classical message & Encoded state & Transferred state & Bob's Outcome & Env's Outcome & \multicolumn{1}{c|}{$p(y_j|x_i)$} \\ \hline\hline
        $x_1$ & $\ket{\phi_1}$ & $\ket{\psi_1}$ & $0$ & $2$ & $p(y_j|x_1)=\delta_{j1}$ \\ \hline
        $x_2$ & $\ket{\phi_2}$ & $\ket{\psi_2}$ & $1$ & $0$ & $p(y_j|x_2)=\delta_{j2}$ \\ \hline
        $x_3$ & $\ket{\phi_3}$ & $\ket{\psi_3}$ & $1$ & $2$ & $p(y_j|x_3)=\delta_{j3}$ \\ \hline
        $x_4$ & $\ket{\phi_4}$ & $\ket{\psi_4}$ & $2$ & $0$ & $p(y_j|x_4)=\delta_{j4}$ \\ \hline
        $x_5$ & $\ket{\phi_5}$ & $\ket{\psi_5}$ & $2$ & $1$ & $p(y_j|x_5)=\delta_{j5}$ \\ \hline \hline
        \multirow{2}{*}{$x_6$} & \multirow{2}{*}{$\ket{\phi_6}$} & \multirow{2}{*}{$\ket{\psi_6}$} & $0$ & $0$ & $p(y_j|x_6)=p\delta_{j6}+(1-p)\delta_{j7}$ \\ \cline{4-6}
         & & & $1$ & $1$ & $p(y_j|x_6)=p\delta_{j6}+(1-p)\delta_{j7}$ \\ \hline \hline
        \multirow{3}{*}{$x_7$} & \multirow{3}{*}{$\ket{\phi_7}$} & \multirow{3}{*}{$\ket{\psi_7}$} & $0$ & $0$ &  $p(y_j|x_7)=p\delta_{j6}+(1-p)\delta_{j7}$ \\ \cline{4-6}
         & & & $1$ & $1$ & $p(y_j|x_7)=p\delta_{j6}+(1-p)\delta_{j7}$ \\ \cline{4-6}
         & & & $2$ & $2$ & $p(y_j|x_7)=\delta_{j7}$ \\ \hline
    \end{tabular}}
\caption{Tabular form of the strategy used by Alice and Bob, under minimally aided by the environment, to generate the channel matrix \(M_7\).}
\label{tab1}
\end{table}
\end{center}
\end{widetext}

{\bf Decoding.} Corresponding to each classical index \(i\), the received joint state of Bob and environment is then given by Eq. \eqref{step2}. Upon receiving the states, both Bob and the environment perform computational basis measurement with effects $\{\ketbra{0}{0},\ketbra{1}{1},\ketbra{2}{2}\}$. Environment, then communicates its result to Bob via a \(\log3\)-bit classical channel. Bob, depending on environment's outcome, declares his variable \(y_j\in Y\). Note that, whenever Bob finds their outputs to be anti-correlated, he perfectly identifies the state \(\ket{\psi_i}\) and hence the encoded classical index \(i\). This happens whenever \(i\in[5]\). However, when Bob finds their output to be correlated he adopts a probabilistic strategy. Specifically, when both their outcomes are either \(0\) or \(1\), he answers \(y_6\) with a probability \(p\) and \(y_7\) with \(1-p\). On the other hand, he always answers \(y_7\) if both of their outcomes are \(2\). A simple observation reveals that this strategy successfully implements \(M_7\). The entire strategy is given in Table \ref{tab1}. 
This establishes that all channels \(\mathcal{N}^{\mathcal{V}}_7\) generates the channel matrix \(M_7\). Now, from 
Lemma \ref{chanmatclas}, we can conclude that \(M_7\notin \mathcal{P}^{7\to7}(\mathcal{Q}_d)\) for all \(d\le6\). Therefore, \(\mathcal{P}^{7\to7}_{\text{env}}(\mathcal{N}_7^\mathcal{V})\nsubseteq\mathcal{P}^{7\to7}(\mathcal{Q}_d)\) for all \(d\le6\). This essentially states that there exist at least one classical information processing task where, under minimal assistance of environment, all \(\mathcal{N}^{\mathcal{V}}_7\) are more useful than \(6\)-dimensional identity quantum channel. This concludes the proof.
\subsection{Proof of Proposition \ref{prop2}}\label{a7}
Clearly the maximum psd rank of all the channel matrices, simulated by a \(d\)-dimensional quantum system, is \(d\). This, along with Eq. (\ref{el22}), further implies the max-monotone \(\Lambda_{\max}\) of all those channel matrices are upper-bounded by \(d\). 

Now, for any non-negative square matrix \(M\) of order \(p\times p\), it is trivial to argue that
\[\Lambda_{\max}(M)=\sum_{k=1}^p\max_{l\in\{1, \cdots, p\}}M_{kl}~~\geq\sum_{i=1}^{p}M_{ii}=\Tr[M].\]
Therefore, \(\forall n\in\mathbb{N}\text{ and }\forall M\in\mathcal{P}^{n\to n}(\mathcal{Q}_d)\),
\begin{align}\label{ep21}
\tr[M]\leq d \implies \mathcal{F}_c(\mathcal{Q}_d)\leq d.
\end{align}
Finally, from Theorem \ref{chanmatduan}, we get \(M_7(p)\in\mathcal{P}^{7\to7}_{env}(\mathcal{N}_7^{\mathcal{V}})\) and \(\tr[M_7(p)]=6+\frac {2p}3\). Therefore, we can evidently conclude \(\mathcal{F}_c^{env}(\mathcal{N}_7^{\mathcal{V}})\geq \tr[M_7(p)] > 6\). This, along with Eq. (\ref{ep21}) (for \(d=6\)), completes the proof. 
\subsection{Proof of Theorem \ref{th2}}\label{a8}
We will prove the theorem in the following two parts:
\subsubsection*{Part 1: Proof of \(\mathcal{F}_c^{SR}(\mathcal{N}^{\mathcal{V}}_7)=\mathcal{F}_c(\mathcal{N}^{\mathcal{V}}_7)\le3\)}\label{pt1}
We will first state few useful lemmas regarding the channel matrices and Classical Transmission Fidelity of any quantum channel.  
\begin{lemma}\label{lm3}
For any quantum channel $\mathcal{M}:\mathcal{L}(\mathbb{C}^{d_A})\!\to\!\mathcal{L}(\mathbb{C}^{d_B})$ and $\forall$ $n,m\in\mathbb{N}$, the set of $n\times m$ channel matrices achievable when the sender and receiver are assisted with an unbounded supply of shared randomness, $
\mathcal{P}_{SR}^{n\to m}\!\big(\mathcal{M}(\mathcal{Q}_{d_A})\big),$ coincides with the convex hull of the unassisted set, i.e,
$$
\mathcal{P}_{SR}^{n\to m}\!\left(\mathcal{M}(\mathcal{Q}_{d_A})\right) = \operatorname{Convhull}\!\left(\mathcal{P}^{n\to m}\!\left(\mathcal{M}(\mathcal{Q}_{d_A})\right)\right).  
$$

\end{lemma}

\begin{proof}
Let us first consider the scenario, where Alice is allowed to send a \(d_A\) level isolated quantum system via the channel \(\mathcal{M}\) assisted by unbounded SR, to Bob. Mathematically, SR can be seen as a set of correlated classical random variables \((K, K)\):  \(\sum_k p_k|k\rangle_A\langle k|\otimes|k\rangle_B\langle k|\), where \(p_k\) is the probability that the random variable \(k\in K\) revealed by both Alice and Bob at an given instant. Then they can implement a pre-decided strategy by communicating a \(d_A\)-level quantum system through \(\mathcal{M}\) and accordingly generates a channel matrix \(P_k\). Therefore, by averaging over the all the random variables we obtain,
\[P=\sum_k p_k P_k.\]
Since, \(\forall k,\quad P_k\in\mathcal{P}^{n\to m}(\mathcal{M}(\mathcal{Q}_{d_A}))\), we can trivially argue that \(P\in\operatorname{Convhull}\!\left(\mathcal{P}^{n\to m}\!\left(\mathcal{M}(\mathcal{Q}_{d_A})\right)\right)\). Hence,
\begin{equation}\label{et2-1}
\mathcal{P}_{SR}^{n\to m}\!\left(\mathcal{M}(\mathcal{Q}_{d_A})\right) \subseteq \operatorname{Convhull}\!\left(\mathcal{P}^{n\to m}\!\left(\mathcal{M}(\mathcal{Q}_{d_A})\right)\right).
\end{equation}
Conversely, consider a channel matrix \(R\in\operatorname{Convhull}\!\left(\mathcal{P}^{n\to m}\!\left(\mathcal{M}(\mathcal{Q}_{d_A})\right)\right)\). Then 
\[R=\sum_i r_i R_i,\]
where, \(\{r_i\in[0,1]:\sum_i r_i=1\}\) is the probability distribution and \(\forall i\quad R_i\) is the extreme channel matrix obtained by sending a \(d_A\)-level quantum system through \(\mathcal{M}\). 

Now, by sharing a SR of the form \(\sum_i r_i|i\rangle_A\langle i|\otimes|i\rangle_B\langle i|\), Alice and Bob can chose a suitable strategy to generate the channel matrix \(R_i\), whenever they locally reveal the random variable \(i\). Hence, the effective channel matrix under such a strategy is 
\[\sum_i r_i R_i=: R\].
This further implies \begin{equation}\label{et20}
\operatorname{Convhull}\!\left(\mathcal{P}^{n\to m}\!\left(\mathcal{M}(\mathcal{Q}_{d_A})\right)\right) \subseteq\mathcal{P}_{SR}^{n\to m}\!\left(\mathcal{M}(\mathcal{Q}_{d_A})\right).    
\end{equation}
Therefore, Eq. (\ref{et2-1}) and (\ref{et20}) together concludes the proof.

\end{proof}
\begin{lemma}\label{lm4}
    For any quantum channel \(\mathcal{M}:\mathcal{L}(\mathbb{C}^{d_A})\mapsto\mathcal{L}(\mathbb{C}^{d_B})\), the Classical Transmission Fidelity does not get enhanced even if the sender and the receiver is aided with unbounded shared randomness, i.e,
    \(
    \mathcal{F}^{SR}_c(\mathcal{M}(\mathcal{Q}_{d_A}))=\mathcal{F}_c(\mathcal{M}(\mathcal{Q}_{d_A}))    
    \).
\end{lemma}
\begin{proof}
We start by noting that from Lemma \ref{lm3}, we have
\begin{eqnarray}\nonumber\mathcal{F}^{SR}_c(\mathcal{M}(\mathcal{Q}_{d_A}))&=&\max_{n}~\max_{P\in\mathcal{P}_{SR}^{n\to n}\!\left(\mathcal{M}(\mathcal{Q}_{d_A})\right)}\tr P\\\nonumber&=&\max_{n}~\max_{P\in\operatorname{Convhull}\left(\mathcal{P}^{n\to n}\!\left(\mathcal{M}(\mathcal{Q}_{d_A})\right)\right)}\tr P\end{eqnarray}

Now, consider the maximum trace is attained by some channel matrix \(\tilde P\in\operatorname{Convhull}(\mathcal{P}^{n\to n}(\mathcal{M}(\mathcal{Q}_{d_A})))\) but \(\tilde P\notin\mathcal{P}^{n\to n}(\mathcal{M}(\mathcal{Q}_{d_A}))\). Now since, \(\tilde P=\sum_kp_k P_k\), where $\forall$ \(k\), \(P_k\in\mathcal{P}^{n\to n}(\mathcal{M}(\mathcal{Q}_{d_A}))\), we have:
\begin{align*}
    \tr \tilde P=\tr \left(\sum_kp_kP_k\right)=\sum_kp_k\tr P_k.
\end{align*}
Therefore,   if \(\tilde P\) attains maximum trace then \(P_k\), $\forall$ \(k\) also attains maximum trace. Now all \(P_k\) are simulated by the channel \(\mathcal{M}\) without the assistance of any shared randomness. This establishes the claim. 
\end{proof}
Now to prove the main claim, we recall from Proposition \ref{prop1} \[\mathcal{P}^{n\to m}(\mathcal{N}^{\mathcal{V}}_7)\subseteq\mathcal{P}^{n\to m}(\mathcal{Q}_3)~ \forall ~n,m\in\mathbb{N}~ \& ~\mathcal{V}\in\mathrm{S}\]
Now from Eq. \eqref{ep21} and Lemma \ref{lm4}, we have \(\mathcal{F}^{SR}_c(\mathcal{N}^{\mathcal{V}}_7)=\mathcal{F}_c(\mathcal{N}^{\mathcal{V}}_7)\le3\). This completes the proof.

\subsubsection*{Part 2: Proof of \(\mathcal{F}_c^{env}(\mathcal{N}^{\mathcal{V}}_7)>\mathcal{F}_c^{PR}(\mathcal{N}^{\mathcal{V}}_7)\)}\label{pt2}
We will start by noting the most general protocol to implement a channel matrix \(P\in\mathcal{P}^{n\to m}(\mathcal{N}_7^{\mathcal{V}})\) when the sender and receiver, along with shared randomness, share any of the channels \(\mathcal{N}_7^{\mathcal{V}}\) and are further assisted by most general \(2\)-input--\(2\)-output non-signaling correlation. Since all such non-signaling correlations can be obtained by first sharing a \(2\)-input--\(2\)-output PR correlation and then by local operation assisted by shared randomness, one can, without loss of generality, restrict the analysis to only shared PR correlation. Furthermore, Lemma \ref{lm4} implies that to optimize the trace, one can only focus on the extreme strategies, i.e, strategies which does not depend on any shared random variables. In the following we will break down such extreme strategies in two parts:

{\bf Encoding:} Consider the message random variable on the sender's (say, Alice) side be \(X:=\{x_1,x_2,\cdots,x_n\}\). Based on her message, Alice performs a boolean function \(f:X\mapsto \{0,1\}\) and puts the bit \(f(x_i)\) as the input to her side of the PR box. Upon receiving the corresponding output \(a_i\in\{0,1\}\), Alice uses a preparation device \(\rho:X\times\{0,1\}\mapsto\mathcal{D}(\mathbb{C}^{7})\) and communicates the state \(\rho(x_i,a_i)\) to the receiver (say Bob), via the channel \(\mathcal{N}^{\mathcal{V}}_7\).

{\bf Decoding:} Upon receiving the state \(\mathcal{N}^{\mathcal{V}}_7(\rho(x_i,a_i))\in\mathcal{D}(\mathbb{C}^{3})\) from Alice, Bob then performs a \(m\)-outcome measurement \(\{M\}_{j=1}^{m}\) with \(M_j\ge0\), $\forall$ \(j\) with \(\sum_jM_j=\mathbb{I}_3\) and obtains the corresponding classical index \(j\in[m]\) as output. Bob then performs a boolean function \(g:[m]\mapsto\{0,1\}\) and inputs \(g(j)\) to his part of the PR box. Upon receiving the output \(b_j\in\{0,1\}\) from the PR box, Bob then performs another function \(d:[m]\times\{0,1\}\mapsto Y\) where \(Y:=\{y_1,y_2,\cdots,y_m\}\) and declares his output variable \(y_j=d(j,b_j)\). Now, due to the preshared PR box they will always have the following
\begin{align}\label{pr}
    a_i\oplus b_j=f(x_i)g(j)
\end{align}
Any advantage that the preshared PR box can provide must come from the correlation in Eq. \eqref{pr}.

Now, we will show that if instead of PR box, Alice and Bob share a random variable \(\lambda\in\{0,1\}\) with a probability distribution \(p(\lambda=0)=p(\lambda=1)=1/2\) and an additional \(1\)-cbit perfect classical channel, they can successfully simulate the same strategy. This follows from the fact that the input output statistics of a PR box can always be simulated by the above resources. However, for sake of completeness we will outline the proof. The strategy to simulate the `PR-box strategy' using the above resources is listed below:

{\bf Encoding:} Alice will follow the same encoding strategy. However, instead of generating the bit \(a_i\) using the PR box, she will now generate it through her part of the shared random variable, i.e, \(a_i=\lambda\). Notably, this works because this unbiased random variable has the same local probability density as the outputs of the PR box. On the other hand, using the \(1\)-cbit perfect classical channel, she will send the bit \(f(x_i)\) to Bob.

{\bf Decoding:} On the decoding end, upon knowing the bit \(f(x_i)\), Bob then generates his bit \(b_j\) as \(b_j=\lambda\oplus f(x_i)g(j)\). The rest of the decoding strategy is exactly same as the previous one. It is very easy to check that, in this way, their corresponding bits will obey the Eq. \eqref{pr}. This establishes the claim.

Therefore, one can always draw the inclusion relation given as \(\mathcal{P}^{n\to m}_{PR}(\mathcal{N}_7^\mathcal{V})\subseteq\mathcal{P}^{n\to m}_{SR}(\mathcal{N}_7^\mathcal{V}+1\text{cbit})\). Now a perfect \(1\)-cbit channel can always be realized by a \(2\)-dimensional quantum identity channel. Therefore, we have
\begin{align*}
&\mathcal{P}^{n\to m}_{PR}(\mathcal{N}_7^\mathcal{V})\subseteq\mathcal{P}^{n\to m}_{SR}(\mathcal{N}_7^\mathcal{V}+1\text{cbit})\subseteq\mathcal{P}^{n\to m}_{SR}(\mathcal{N}_7^\mathcal{V}+\mathcal{Q}_2)\\
\Rightarrow &\mathcal{P}^{n\to m}_{PR}(\mathcal{N}_7^\mathcal{V})\subseteq\mathcal{P}^{n\to m}_{SR}(\mathcal{Q}_3+\mathcal{Q}_2)~[\text{From Proposition \ref{prop1} }]\\
\Rightarrow &\mathcal{P}^{n\to m}_{PR}(\mathcal{N}_7^\mathcal{V})\subseteq\mathcal{P}^{n\to m}_{SR}(\mathcal{Q}_5)~[\text{No hypersignaling principle} \cite{dall2017no}]\\
\Rightarrow &\mathcal{F}_c^{PR}(\mathcal{N}_7^\mathcal{V})\le\mathcal{F}_{c}^{SR}(\mathcal{Q}_5)\le5~[\text{From Eq. \eqref{ep21} and Lemma \ref{lm4}}]\\
\end{align*}
This completes the proof.

Combining both Part \hyperref[pt1]{1}, Part \hyperref[pt2]{2} and Proposition \ref{prop2}, we get \(\mathcal{F}_c^{env}(\mathcal{N}^{\mathcal{V}}_7)>\mathcal{F}_c^{PR}(\mathcal{N}^{\mathcal{V}}_7)\ge \mathcal{F}_c^{SR}(\mathcal{N}^{\mathcal{V}}_7)\). This finally completes the proof of Theorem 2.

\subsection{Proof of Theorem \ref{gen}}\label{a9}
To prove the theorem, let us first consider a generic class of channel matrices \(M_{k_d}\) of order \(k_d\times k_d\), with \(k_d=d^2-1\), of the form:
\begin{align}\label{et21}
\nonumber M_{k_d}&=M_{\sum}\oplus\mathbb{I}_{\tilde{d}},\text{ where, }\tilde{d}=d(d-1),\text{ and }\\\nonumber\\
 (M_{\sum})_{i,j} &=\begin{cases}
       ~~~0 \quad\quad\text{for }1\leq i<j\leq (d-1),\\\frac 2{i(i+1)}\quad\text{for }j=1,\\\frac 1{i(i+1)}\quad\text{for }1<j<i\leq(d-1),\\\frac j{(j+1)}\quad~\text{for }1<i=j\leq (d-1).
   \end{cases}
\end{align}
\begin{lemma}\label{sl3}
     The PSD rank of the channel matrix \(M_{k_d}\) is \(k_d\).
\end{lemma}
\begin{proof}
    Since \(M_{k_d}\) is a block-diagonal matrix, using Lemma \ref{sl1} and Eq.(\ref{el23}) we can immediately conclude 
    \begin{align}\label{et22}
       \nonumber \text{rank}_{\text{psd}}(M_{k_d})&=\text{rank}_{\text{psd}}(M_{\sum})+\text{rank}_{\text{psd}}(\mathbb{I}_{\tilde{d}})\\&=\text{rank}_{\text{psd}}(M_{\sum})+\tilde{d}.
    \end{align}
    Further note that, \(M_{\sum}\) is a lower-triangular matrix of order \((d-1)\times(d-1)\), specifically:
    \[M_{\sum}=\begin{pmatrix}
        1&0&0&0&\cdots\\
        \frac 13&\frac 23&0&0&\cdots\\
        \frac 16&\frac 1{12}&\frac 34&0&\cdots\\
        \frac 1{10}&\frac 1{20}&\frac 1{20}&\frac 45&\cdots\\
        ..&..&..&..&\cdots\\
        ..&..&..&..&\cdots
    \end{pmatrix}.\]
    Therefore, by using Lemma \ref{sl2} we have \[\text{rank}_{\text{psd}}(M_{\sum})\geq d-1.\]
    This along with Eq. (\ref{et22}), further implies 
    \[\text{rank}_{\text{psd}}(M_{k_d})\geq (d-1)+\tilde{d}=d^2-1.\]
    We will now conversely show that \(\text{rank}_{\text{psd}}(M_{k_d})\leq d^2-1\) by simply providing a quantum strategy to simulate the channel matrix \(M_{k_d}\).
    \medskip
    
    \textbf{Encoding:} Given a \(k_d=d^2-1\) dimensional perfect quantum channel, Alice can encode her input random variable \(\{x_i\}_{i=0}^{k_d-1}\) using the quantum states \(\{\ket{i}\}_{i=0}^{k_d-1}\in\mathbb{C}^{k_d}\) and send them to Bob.\par
    \medskip
    \textbf{Decoding:} Bob will then perform a computational basis measurement \(\{\ketbra{i}{i}\}_{i=0}^{k_d-1}\). 
    For the first \(d-1\) outcomes he will use a cunningly chosen local random variables to generate the matrix  \(M_{\sum}\). 
    In particular, his output will be \(y_0\), whenever the projector \(\ketbra{0}{0}\) clicks. On the other hand, for the clicking of the projector \(\ketbra{k}{k},~(0< k\leq d-2)\), Bob will answer the random variables \(\{y_{j_k}\}_{j_k=0}^{k}\), with the probabilities \(\{\frac 2{(k+1)(k+2)},\frac 1{(k+1)(k+2)},\cdots,\frac 1{(k+1)(k+2)},\frac {k+1}{k+2}\}\) respectively.\par
\noindent
     Additionally, for the last \(d(d-1)\) projectors \(\{\ketbra{k}{k}\}_{k=d-1}^{k_d-1}\), Bob's output will be simply \(y_k\).\par\noindent
     This completes the proof.
\end{proof}
Coming back to the proof of the main theorem, in the following we will show that all the channels of the form \(\mathcal{N}^{\mathcal{V}}_{d^2-1}\) will be able to generate the channel matrix \(M_{k_d},~\forall d\geq 3\), in spite of possessing suboptimal EACC. 

\textbf{Encoding:} Let us first consider a set of orthogonal basis \(\{\ket{\psi_k}\}_{k=0}^{d^2-2}\in\mathbb{C}^d\otimes\mathbb{C}^d\), spanning the subspace \(\mathrm{S}_d\perp\{\ket{\phi^+_d}=\frac 1{\sqrt{d}}\sum_{i=0}^{d-1}\ket{i}\otimes\ket{i}\}\) (see Eq. (\ref{duanspace}) in the main text). In particular, the first \((d-1)\) states are entangled:
\begin{equation}\label{et23}
\ket{\psi_k}=\frac 1{\sqrt{k+2}}\ket{\phi^+_{k+1}}-\frac {\sqrt{k+1}}{\sqrt{k+2}}\ket{k+1}\otimes\ket{k+1},
\end{equation}
where \(k\in\{0,1,\cdots, d-2\}\) and the rests are products of the form
\begin{align}\label{et24}
    \nonumber\ket{\psi_k}&=\ket{i}\otimes\ket{j},\quad\text{with }i\neq j~\&~i,j\in\{0,\cdots,d-1\}\\\text{where, }k&=
    \begin{cases}
        i\cdot (d-1) + j + (d-1) \text{ for }i>j\\
        i\cdot (d-1) + (j-1) + (d-1)\text{ for }i<j
        \end{cases}
\end{align}
\noindent
Now, Alice encodes her input random variables \(\{x_i\}_{i=0}^{k_d-1}\) in \(k_d=d^2-1\) dimensional orthogonal quantum states \(\{\ket{\xi_i}\}_{i=0}^{d^2-2}\), such that the isometry \(\mathcal{V}_{d^2-1}\) as the following
\[\mathcal{V}_{d^2-1}\ket{\xi_i}=\ket{\psi_i},~\forall i\in\{0, 1, \cdots, d^2-2\},\]
where, the states \(\ket{\psi_i}\) are shared between the receiver Bob and the environment. 

\textbf{Decoding:} To generate the output random variables \(\{y_j\}_{j=0}^{d^2-2}\), Bob will perform a computational basis measurement \(\{|i\rangle_B\langle i|\}_{i=0}^{d-1}\) on his local constituents. Additionally, he will be informed the outcome of the measurement \(\{|{i}\rangle_E\langle{i}|\}_{i=0}^{d-1}\), independently performed on the environments side. The decoding structure readily identifies the setting as a minimal assistance one. We will denote the outcome obtained by Bob as \(b\in\{0,1,\cdots,d-1\}\), while \(e\in\{0,1,\cdots,d-1\}\) as the \(\log d\)-bit classical information from the environment. 

Now, Bob will output \(y_j\), where \(j=b\cdot (d-1)+e+(d-1)\) when \(b>e\) and \(j=b\cdot (d-1)+(e-1)+(d-1)\) when \(b<e\). Note that, in either cases \(j\in\{d-1,d,\cdots,d(d-1)\}\). It can be trivially argued from Eq. (\ref{et23}) and (\ref{et24}) that such an instance of clicking the projectors \(\ketbra{b}{b}\otimes\ketbra{e}{e}\) can only happen when \(\ket{\psi_j}=\ket{b}_B\otimes\ket{e}_E\) where \(b\neq e\). That is, the state sent by Alice is indeed \(\ket{\xi_j}\). Therefore, they can successfully simulate the channel matrix \(\mathbb{I}_{\tilde{d}}\), where \(\tilde{d}=d(d-1)\).

On the other hand, consider the case where \(b=e\), then the state distributed between Bob and the environment must be among \(\{\ket{\psi_0},\cdots,\ket{\psi_{d-2}}\}\). In those cases, Bob will simply output the index \(y_0\), whenever \(b=e\in\{0,1\}\) and output \(y_{b-1}\) when \(b=e\notin\{0,1\}\). 

As an illustration, consider the state \(\ket{\psi_k}\) as in Eq. (\ref{et21}), the projectors \(\ketbra{0}{0}\otimes\ketbra{0}{0}\) and \(\ketbra{1}{1}\otimes\ketbra{1}{1}\) clicks with a probability 
\[|\langle 00|\psi_k\rangle|^2+|\langle 11|\psi_k\rangle|^2=\frac 1{k+2}\times \frac 2{k+1}=:p(y_0|x_k),\]
which mimics the first column of the \((k+1)^{th}\) row. Similarly, for the \(l^{th}\) column of the same row we obtain
\[|\langle ll|\psi_k\rangle|^2=\frac 1{k+2}\times\frac 1{k+1}=:p(y_{l-1}|x_k),~\forall l\in\{2,\cdots,k\}\]
and finally to conclude the proof, for the \((k+1)^{th}\) column 
\[|\langle (k+1)(k+1)|\psi_k\rangle|^2=\frac {k+1}{k+2}=:p(y_k|x_k).\]

\end{document}